\documentclass[letterpaper]{article} 
\usepackage{aaai25}  
\usepackage{times}  
\usepackage{helvet}  
\usepackage{courier}  
\usepackage[hyphens]{url}  
\usepackage{graphicx} 
\usepackage{natbib}  
\usepackage{caption} 
\usepackage{algorithm}
\usepackage{algorithmic}

\usepackage{listings}
\DeclareCaptionStyle{ruled}{labelfont=normalfont,labelsep=colon,strut=off} 
\floatstyle{ruled}
\newfloat{listing}{tb}{lst}{}
\floatname{listing}{Listing}
\title{Equilibria of the Colonel Blotto Games with Costs}
\author {
    Stanisław Kaźmierowski
}
\affiliations{
    University of Warsaw\\
    s.kazmierowski@uw.edu.pl
}

\usepackage{enumitem}
\usepackage{amsmath}
\usepackage{amsfonts}
\usepackage{amsthm}
\usepackage{bm}
\usepackage{mathtools}
\usepackage{tabularx}

\usepackage{xcolor}

\newtheorem{theorem}{Theorem}
\newtheorem{proposition}{Proposition}
\newtheorem{corollary}{Corollary}

\theoremstyle{definition}
\newtheorem{observation}{Observation}
\newtheorem{hypothesis}{Hypothesis}
\newtheorem{definition}{Definition}
\newtheorem{example}{Example}

\DeclareMathSymbol{\shortminus}{\mathbin}{AMSa}{"39}

\newcommand{\sign}{\mathrm{sign}}
\newcommand{\playerA}{\mathsf{A}}
\newcommand{\playerB}{\mathsf{B}}
\newcommand{\thisPlayer}{P}
\newcommand{\thatPlayer}{{\shortminus P}}
\newcommand{\examplePlayer}{\thisPlayer \!\in \!\left\{ \playerA, \playerB \right\}}
\newcommand{\thisRes}{D^\thisPlayer}

\newcommand{\aRes}{D^\playerA}
\newcommand{\bRes}{D^\playerB}
\newcommand{\fieldsSet}{\left[n\right]}
\newcommand{\pure}{\bm{s}}
\newcommand{\thisPure}{\pure^\thisPlayer}
\newcommand{\thatPure}{\pure^\thatPlayer}
\newcommand{\aPure}{\pure^\playerA}
\newcommand{\bPure}{\pure^\playerB}

\newcommand{\thisPureSet}{S^\thisPlayer}

\newcommand{\aPureSet}{S^\playerA}
\newcommand{\bPureSet}{S^\playerB}

\newcommand{\thisPayoff}{\Pi^\thisPlayer_{\$} \!}

\newcommand{\mixed}{\bm{\xi}}
\newcommand{\thisMixed}{\mixed^\thisPlayer}

\newcommand{\aMixed}{\mixed^\playerA}
\newcommand{\bMixed}{\mixed^\playerB}
\newcommand{\mixedProfile}{\left(\aMixed, \bMixed\right)}
\newcommand{\mixedPair}{\aMixed, \bMixed}
\newcommand{\thisMixedSet}{\Delta\left(\thisPureSet\right)}

\newcommand{\thisMixedPayoff}{\thisPayoff \mixedProfile}

\newcolumntype{b}{X}
\newcolumntype{s}{>{\hsize=.5\hsize}X}

\begin{document}

\maketitle

\begin{abstract}
This paper studies a generalized variant of the Colonel Blotto game, referred to as the Colonel Blotto game with costs. Unlike the classic Colonel Blotto game, which imposes the \textit{use-it-or-lose-it} budget assumption, the Colonel Blotto game with costs captures the strategic importance of costs related both to obtaining resources and assigning them across battlefields. We show that every instance of the Colonel Blotto game with costs is strategically equivalent to an instance of the zero-sum Colonel Blotto game with one additional battlefield. This enables the computation of Nash equilibria of the Colonel Blotto game with costs in polynomial time with respect to the game parameters: the number of battlefields and the number of resources available to the players.
\end{abstract}


\section{Introduction}

In the competitive landscape of the global electric vehicle (EV) market, Tesla and Volkswagen are two major players making strategic decisions about increasing and allocating production capacity across multiple key markets. Each market has unique characteristics, such as production costs, market potential, and regulations, that influence both the potential return on investment and the costs associated with expanding production capacity. Furthermore, the costs of increasing production capacity in any given market may differ among companies due to factors like existing infrastructure, supply chain efficiency, and local partnerships. This creates a dual challenge: first, deciding how much to invest in expanding overall production capacity, and second, determining how to allocate this capacity across different markets to maximize competitive advantage. Strategic planners must carefully weigh these considerations to decide how to globally increase and distribute production capacity. The central questions are: How should a company balance its investments in expanding capacity with the costs and benefits of allocating it to specific markets? What is the optimal strategy for assigning the increased capacity across markets with varying cost structures? Finally, what are the total costs that each company incurs in pursuing this global expansion?

Similarly, in the competitive pharmaceutical industry, companies like Pfizer and Johnson \& Johnson are engaged in a strategic race to develop new drugs. They must decide how much to invest in expanding their R\&D capacity---such as hiring scientists and expanding lab facilities ---- and how to allocate these resources to different research areas such as oncology, immunology, and neurology. Each area has its own set of challenges, including market potential, scientific difficulty, and regulatory requirements, which affect both the cost and the potential return on investment. In addition, the cost of conducting research in these areas can vary between companies due to factors such as existing expertise, access to specialized equipment, and previous partnerships with research institutions. The key strategic questions are: How should a company increase its R\&D capacity? How should it allocate this capacity across various research areas to balance risk, cost, and reward? What are the overall costs for each company in this R\&D competition?

Strategic decisions regarding the aforementioned scenarios share similar characteristics. First, they involve allocating resources across multiple contests, where strategic planners aim to maximize the sum of utilities across all contests. Second, both obtaining and assigning resources come at a cost. A strategic planner aiming to maximize the outcome of her decisions needs to consider both the anticipated outcomes of contests as well as the total costs incurred. The Colonel Blotto game~\cite{borel} is a well-known model that captures the strategic problem of allocating limited resources across multiple contests. However, in its classic formulation, the Colonel Blotto game imposes a sunk-cost assumption, i.e., the budgets of competing parties are fixed, and there are no costs related to either obtaining resources or assigning them to specific battlefields. Since the sunk-cost assumption does not capture the scenarios mentioned above, we study the problem of computing equilibrium strategies in a generalized variant of the Colonel Blotto game, called the Colonel Blotto game with costs.

\section{Colonel Blotto Game}
The Colonel Blotto game was first introduced by Borel in~\citeyear{borel}. In its original formulation, two colonels are tasked with simultaneously assigning their limited number of troops across the set of battlefields. Each colonel wins all the battlefields where he assigns strictly more resources than the opponent, and the remaining battlefields are tied. In the original formulation of the model, all battlefields are equally valuable, and the payoff to a colonel is the difference between the number of won and lost battlefields. Although the original motivation for introducing the Colonel Blotto game was military~\cite{borel, blotto_military_1, blotto_military_2}, the model has been since applied for analyzing competitions in various scenarios, such as politics~\cite{blotto_elections_3, blotto_elections_4, blotto_elections_5, blotto_elections_1}, network security~\cite{network_security_blotto_1, network_security_blotto_3, network_security_blotto_2, ChiaChuang11}, and marketing~\cite{blotto_marketing_3, blotto_marketing_2, blotto_marketing_1}. A simple description of this game combined with a rich structure and many possible applications resulted in the ongoing interest in solutions of the Colonel Blotto game of researchers in multiple areas. Examples of these fields are operations research~\cite{blotto_military_1, BealeHeselden62, operations_research_blotto, blotto_military_2}, economics~\cite{hart, cont_blotto_solved, kvasov, copula_1}, and computer science~\cite{winning_blotto_auditing, duels_to_battlefields_blotto, blotto_favourism, fast_2023_blotto}.

The first stream of research concerning solutions of the Colonel Blotto game aims to obtain (at least partial) characterization. Although in most formulations the game is zero-sum, the exponential number of strategies with respect to the model parameters makes the problem of characterization of optimal strategies hard, and most results were obtained under restrictive assumptions that narrow their possible applications. One of such restrictive assumptions is relaxing the integrality constraint of the problem, i.e., allowing the players to split their troops, resulting in the \textit{continuous} variant of the model. Working under this assumption, Borel and Ville~(\citeyear{blotto_3_battlefields}) solved the game for three battlefields and the same number of resources of both players. Their results were extended by Gross and Wagner~(\citeyear{blotto_Gross1950ACC}), who obtained a general solution under the symmetric budget of both players. Roberson~(\citeyear{cont_blotto_solved}) extended the knowledge on continuous variants of the game, providing solutions for the general case both in terms of the number of battlefields and asymmetric budgets of players. These results were further generalized by Kovenock and Roberson~(\citeyear{KovenockRoberson21}), who solved the variant allowing for different valuations across battlefields. 
As for the discrete variant of the problem, Hart~(\citeyear{hart}) provided partial characterization of solutions of the budget-symmetric model as well as for some of the budget-asymmetric cases for a general number of battlefields.
As for today, the general characterization of the equilibrium set in Colonel Blotto games with discrete strategy spaces has remained elusive. 

Difficulties in characterizing solutions of the more general variants of the Colonel Blotto game (e.g. different values across battlefields or different values of battlefields across players) resulted in increased attention regarding the problem of computing or approximating solutions of such models. Beale and Heselden~(\citeyear{BealeHeselden62}) introduced a method of approximating solutions of the Colonel Blotto game, based on a new game involving the mean numbers of forces allocated to each battlefield in the original game. A more generalized multiplier method, allowing for computing solutions of a broader class of games, including the Colonel Blotto game was introduced by Penn~(\citeyear{Penn71}). Combining the idea multiplier method with the fictitious play resulted in an implementation for computing the Colonel Blotto game solutions, described by Eisen and Le Mat~(\citeyear{EisenLeMat68}).

Recently, Ahmadinejad et al.~(\citeyear{duels_to_battlefields_blotto}) proposed a first method of computing solutions for the Colonel Blotto game in polynomial time with respect to the number of resources of both players and the number of battlefields. The proposed method utilized the idea of separation oracle and was proved applicable to a broader class of games satisfying the bilinearity of payoffs, i.e., the payoff to a player is the sum of payoffs across single battlefields. Although considered a breakthrough because of the guarantees of polynomial runtime, the method proposed by Ahmadinejad et al.~(\citeyear{duels_to_battlefields_blotto}) uses the ellipsoid method, which makes it insufficient for larger model parameters. Behnezhad et al.~(\citeyear{fast_2023_blotto}) addressed this problem, providing a first polynomial-size LP formulation of the problem, resulting in a simpler and significantly faster algorithm. Furthermore, they showed that their LP representation is asymptotically tight. Their method, based on the idea of representing a mixed strategy of a player as a flow in the corresponding layered graph, is applicable to every zero-sum Colonel Blotto game satisfying the bilinearity of payoffs and the zero-sum property at every battlefield, and it is considered the state-of-the-art algorithm for computing solutions of the Colonel Blotto game. 

As for the variants of the Colonel Blotto game that allow for capturing the costs of either obtaining or assigning resources, Roberson and Kvasov~(\citeyear{kvasov}) considered a Colonel Blotto game with continuous resources and linear costs. They show that if the level of asymmetry between the players’ budgets is below a threshold, then there exists a one-to-one mapping between the solutions of the Colonel Blotto game with costs and the solutions of the Colonel Blotto game. We are not aware of any other results on solutions of the Colonel Blotto game with costs outside of the work of Roberson and Kvasov~(\citeyear{kvasov}).

\paragraph{Our contribution:} We study a generalized variant of the Colonel Blotto game, called the \textit{Colonel Blotto game with costs}, which captures both global costs of obtaining the resources as well as battlefield-specific costs of assigning resources to the given battlefield. Although the game in question is not zero-sum (or even strictly competitive), using the notion of strategical equivalence, we show that the Colonel Blotto game with costs satisfies the interchangeability of Nash equilibria. Furthermore, we prove that every instance of the Colonel Blotto game with costs is strategically equivalent to a corresponding variant of the standard Colonel Blotto game where players have fixed budgets. By applying the results of Behnezhad et al.~(\citeyear{fast_2023_blotto}) about solution computation in the standard Colonel Blotto games with fixed budgets, we show that the problem of computation of Nash equilibria of the Colonel Blotto game with costs can be expressed using a linear program with at most $\Theta (D^2 \cdot n)$ constraints and $\Theta (D^2 \cdot n)$ variables, where $n$ is the number of battlefields and $D$ is the maximum number of resources that any player is allowed to assign. Lastly, we present our experimental results and discuss their consistency with theoretical results for the continuous resources setting of Roberson and Kvasov~(\citeyear{kvasov}).

\section{Model}\label{sec:model} 

This section introduces the Colonel Blotto game with discrete resources and costs, a generalization of the classic Colonel Blotto game that captures the strategic costs of assigning resources over a set of battlefields. 
The Colonel Blotto game with costs models competition between two players denoted by $\playerA$ and $\playerB$. Each player has a number of available discrete resources, e.g. military units or coins, denoted by $\aRes$ and $\bRes$, respectively. Players compete on the set $\fieldsSet =\left \{1,2,\ldots\,n\right\}$ of $n$ battlefields, where $n \geq 2$. A strategy of player $P \in \{\playerA, \playerB\}$ is an assignment of at most $D^P$ resources across the set of battlefields. Formally, the set of strategies of player $\thisPlayer \in \left\{ \playerA, \playerB \right\}$ is
\begin{equation}
\label{pure_strategy}
    S(D^P, n) = \left\{ \pure \in \mathbb{N}^n : \sum^n_{i=1} s_i \leq \thisRes \right\}.
\end{equation}
Elements of $S(D^P, n)$ are partial assignments, as the players are not required to use all of their available resources. When the parameters are clear from the context, the set of partial assignments is denoted $S^P$.

Given a strategy profile $(\aPure, \bPure)$, the payoffs to the players from battlefield $i \in \left[n\right]$ are defined as
\begin{align*}
\label{pure_single_payoff}
    & u_i^{\playerA} \left(s^{\playerA}_i, s^{\playerB}_i\right) = \hphantom{-} v_i \left(s^{\playerA}_i, s^{\playerB}_i\right) -c_i^{\playerA}\left(s^{\playerA}_i\right), \\
    & u_i^{\playerB} \left(s^{\playerA}_i, s^{\playerB}_i\right) = - v_i\left(s^{\playerA}_i, s^{\playerB}_i\right) - c_i^{\playerB}\left(s^{\playerB}_i\right),
\end{align*}
where $v_i : \mathbb{N} \times \mathbb{N} \longrightarrow  \mathbb{R}$ is a valuation function associated with the battlefield $i$, and $c_i^{\playerA}, c_i^{\playerB}: \mathbb{N} \longrightarrow  \mathbb{R}$ are non-decreasing \textit{assignment cost functions} that capture the costs of assigning a given number of resources to the given battlefield by each of the players.

The payoff to player $\thisPlayer \in \left\{ \playerA, \playerB \right\}$ from strategy profile $(\bm{s}^P, \bm{s}^{\shortminus P})$ is
\begin{equation}
\label{pure_payoffs}
    \pi_{\$}^\thisPlayer \! \left( \thisPure, \thatPure \right) = \sum_{i=1}^n u_i^{P}\left(s^P_i, s^{\shortminus P}_i\right) - g^{\thisPlayer} \left(\sum_{i=1}^n s_i^{\thisPlayer}\right),
\end{equation}
where $g^{P}$ is a non-decreasing \textit{obtainment cost} function that describes the cost of obtaining the total number of resources assigned across all the battlefields. 

\begin{definition}[Colonel Blotto game with costs]
For a given tuple $(\aRes, \bRes, n, \bm{v}, \bm{c}^{\playerA},\bm{c}^{\playerB}, g^{\playerA}, g^{\playerB})$, a strategic game with strategy sets defined by \eqref{pure_strategy} and payoffs defined by \eqref{pure_payoffs} is denoted with $\mathcal{B}_{\$}(\aRes, \bRes, n, \bm{v}, \bm{c}^{\playerA},\bm{c}^{\playerB} , g^{\playerA}, g^{\playerB}) \equiv \mathcal{B}_{\$}$ and called the \textbf{Colonel Blotto game with costs}. 
\end{definition}

The players are allowed to make randomized choices. A \emph{mixed strategy} of player $\examplePlayer$ is a probability distribution on $S^P$. 
Given a non-empty set $X$, let $\Delta(X)$ denote the set of all probability distributions on~$X$. 
The expected payoff to player $\examplePlayer$ from a pair of mixed strategies $(\mixedPair) \in \Delta(S^{\playerA}) \times \Delta(S^{\playerB})$ is equal to
\begin{equation*}
    \thisMixedPayoff = \sum\limits_{(\bm{x}, \bm{z}) \in \aPureSet \times \bPureSet} \xi^\playerA_{\bm{x}} \xi^\playerB_{\bm{z}} \pi_{\$}^\thisPlayer \! \left(\bm{x}, \bm{z} \right),
\end{equation*}
where given the probability distribution $\xi^{P}$, $\xi^{P}_{\bm{x}}$ is the probability associated with pure strategy $\bm{x}$. We assume that the players make their decisions ``simultaneously'', i.e., each player chooses her strategy without observing the choice of the opponent. We are interested in mixed strategy Nash equilibria (MSNE) of the game, i.e., mixed strategy profiles $(\mixedPair) \in \Delta(S^{\playerA}) \times \Delta(S^{\playerB})$ such that no player can improve her expected payoff by changing her strategy unilaterally.  Formally, for each $\examplePlayer$ and all $\bm{\zeta} \in \thisMixedSet$,
\begin{equation*}
\thisMixedPayoff \geq \thisPayoff\left(\bm{\zeta},\bm{\xi}^\thatPlayer\right),    
\end{equation*}
where $\thatPlayer$ denotes the player other than $\thisPlayer$ and $(\bm{\zeta},\bm{\xi}^\thatPlayer)$ is the strategy profile
obtained from $(\mixedPair)$ by replacing $\thisMixed$ with $\bm{\zeta}$. 

\section{Analysis}
In this section, we present one of two key ideas of this paper, that is to introduce an artificial zero-sum game, $\mathcal{B}_{0}$, that, given the same set of parameters, has the same set of Nash equilibria as the Colonel Blotto game with costs, $\mathcal{B}_{\$}$. The relation between the two games provides insight into the equilibria properties of the Colonel Blotto game with costs.

For the earlier introduced set of parameters, let $\mathcal{B}_{0}(\aRes, \bRes, n, \bm{v}, \bm{c}^{\playerA},\bm{c}^{\playerB}, g^{\playerA}, g^{\playerB}) \equiv \mathcal{B}_{0}$ denote a two-player strategic game with strategy sets of both players defined by~\eqref{pure_strategy}. The payoff to player $P \in \{\playerA, \playerB\}$ in the game $\mathcal{B}_{0}$ under strategy profile $(\bm{s}^P, \bm{s}^{\shortminus P})$ is
\begin{align*}
    &\pi^\thisPlayer_{0} \! \left( \thisPure, \thatPure \right) = \\
     \pi^\thisPlayer_{\$} & \! \left( \thisPure, \thatPure \right) + \sum_{i=1}^n c_i^{\shortminus P}\left(s^{\shortminus P}_i\right) + g^{\shortminus P} \left(\sum_{i=1}^n s_i^{\shortminus P}\right).
\end{align*}

Game $\mathcal{B}_{0}$ is obtained from the corresponding game $\mathcal{B}_{\$}$ by adding to the payoff of each player the sum of costs (both general and battlefield-specific) carried by her opponent, resulting in the zero-sum property.
The following proposition characterizes the relation between the Nash equilibria of $\mathcal{B}_{\$}$ and $\mathcal{B}_{0}$ given the same set of parameters.
\begin{proposition}~\label{prop:equiv}
For the same set of parameters $(\aRes, \bRes, n, \bm{v}, \bm{c}^{\playerA},\bm{c}^{\playerB}, g^{\playerA}, g^{\playerB})$, games $\mathcal{B}_{\$}$ and $\mathcal{B}_{0}$ have the same set of Nash equilibria.
\end{proposition}

\begin{proof}
Fix player $\examplePlayer$, two strategies $\bm{s}^P, \bm{t}^P \in S^P$ and a strategy $s^{\shortminus P} \in S^{-P}$ of player $-P$. It holds, that
\begin{align*}
    &\pi^\thisPlayer_{0} \! \left( \thisPure, \thatPure \right) - \pi^\thisPlayer_{0} \! \left( \bm{t}^P, \thatPure \right) = \\ &\pi^\thisPlayer_{\$} \! \left( \thisPure, \thatPure \right) - \pi^\thisPlayer_{\$} \! \left( \bm{t}^P, \thatPure \right).
\end{align*}
Hence, the two games are strategically equivalent in the sense that for every player $\examplePlayer$, every two strategies $\bm{s}^P, \bm{t}^P \in S^P$ and every strategy of the opponent $s^{\shortminus P} \in S^{-P}$, it holds, that 
\begin{align*}
    &\pi^\thisPlayer_{0} \! \left( \thisPure, \thatPure \right) \geq \pi^\thisPlayer_{0} \! \left( \bm{t}^P, \thatPure \right) \iff \\ &\pi^\thisPlayer_{\$} \! \left( \thisPure, \thatPure \right) \geq \pi^\thisPlayer_{\$} \! \left( \bm{t}^P, \thatPure \right).
\end{align*}
As a consequence of Lemma~1 of Moulin and Vial~(\citeyear{equivalence_0}), the two games have the same set of Nash equilibria. 
\end{proof}

\subsection{Equilibria Properties}
This subsection provides a discussion of the properties of Nash equilibria of the Colonel Blotto game with costs. 
First, we give two corollaries that follow directly from Proposition~\ref{prop:equiv}. A strategic two-player game satisfies the interchangeability of Nash equilibria~\cite{interchange_1, interchange_2} if for every two (possibly mixed) strategy profiles $(\mixedPair), (\bm{\sigma}^{\playerA}, \bm{\sigma}^{\playerB}) \in \Delta(S^{\playerA}) \times \Delta(S^{\playerB})$ that are Nash equilibria of game $\mathcal{B}_{\$}$, it holds that strategy profiles $(\bm{\xi}^{\playerA}, \bm{\sigma}^{\playerB})$ and $(\bm{\sigma}^{\playerA}, \bm{\xi}^{\playerB})$ are also Nash equilibria of game $\mathcal{B}_{\$}$. 
As every finite two-player zero-sum game satisfies the interchangeability of Nash equilibria, the following corollary holds.
\begin{corollary}\label{cor:interchangeability}
The Colonel Blotto game with costs satisfies the interchangeability of Nash equilibria.
\end{corollary}
Note that one of the consequences of the interchangeability of Nash equilibria in a two-player finite game is that every player has a set of equilibrium strategies, and every pair of equilibrium strategies is a Nash equilibrium of the game. In the case of zero-sum games, equilibria strategies are called optimal strategies. As the Colonel Blotto game with costs is not zero-sum, we use the term \textit{equilibrium strategy}.
Since the set of optimal strategies in every finite two-player zero-sum game is convex, the following corollary is also a direct consequence of Proposition~\ref{prop:equiv}.
\begin{corollary}\label{cor:convex}
For each player $\thisPlayer$ the set of equilibrium strategies of player $\thisPlayer$ is convex in every Colonel Blotto game with costs.
\end{corollary}

We illustrate the consequences of these corollaries in the following example.
\begin{example}
Consider an instance $\mathcal{B}_{\$}^*$ of the Colonel Blotto game with costs, where
\begin{itemize}
    \item the players compete on two battlefields,
    \item each player can assign at most 2 resources,
    \item each battlefield is won by the player assigning strictly more resources to it (and tied otherwise),
    \item obtaining a unit of resources yields a cost of 1 for each of the players,
    \item assigning resources across battlefields yields no costs. 
\end{itemize}
Formally, 
$\mathcal{B}_{\$}^* \equiv \mathcal{B}_{\$}(2, 2, 2, \bm{v} = (\bm{\sign}), \bm{0}, \bm{0}, g^{\playerA}, g^{\playerB})$, where
    \begin{equation*}
        g^{\playerA}(\bm{s}) = g^{\playerB}(\bm{s}) = s_1 + s_2.
    \end{equation*}
The strategy sets of both players are
\begin{equation*}
    S^{\playerA} = S^{\playerB} = \{(0,0),(0,1),(1,0),(1,1),(0,2),(2,0)\},
\end{equation*}
and it can be easily calculated, that every strategy in
\begin{equation*}
    S^* = \{(0,0),(0,1),(1,0),(1,1)\},
\end{equation*}
is an equilibrium strategy of $\mathcal{B}_{\$}^*$, while remaining pure strategies $(0,2), (2,0)$ are never used in any equilibrium strategy. A consequence of Corollary~\ref{cor:convex} is that every mixed strategy with support restricted to any subset of $S^*$ is also an equilibrium strategy of $\mathcal{B}_{\$}^*$. Figure~\ref{fig:payoff_matrix} present the payoff matrix of game $\mathcal{B}_{\$}^*$. Parts of the payoff matrix corresponding to pure strategies outside of $S^*$ are shadowed, to focus attention on the strategies in $S^*$.
\begin{figure}[H]
\begin{center}
\begin{tabularx}{0.95\columnwidth}{ 
   >{\centering\arraybackslash}X |
   >{\centering\arraybackslash}X |
   >{\centering\arraybackslash}X |
   >{\centering\arraybackslash}X |
   >{\centering\arraybackslash}X |
   >{\centering\arraybackslash}X |
   >{\centering\arraybackslash}X }
 & (0,0)  & (0,1) & (1,0) & (1,1) & \textcolor{gray}{(0,2)} & \textcolor{gray}{(2,0)}  \\ 
\hline
$(0,0)$ & $\mathbin{\phantom{\shortminus}}0,\mathbin{\phantom{\shortminus}}0$ & $\shortminus1,\mathbin{\phantom{\shortminus}}0$ & $\shortminus1,\mathbin{\phantom{\shortminus}}0$ & $\shortminus2,\mathbin{\phantom{\shortminus}}0$& \textcolor{gray}{$\shortminus1,\shortminus1$} & \textcolor{gray}{$\shortminus1,\shortminus1$} \\
\hline
(0,1) & $\mathbin{\phantom{\shortminus}}0,\shortminus1$ & $\shortminus1,\shortminus1$ & $\shortminus1,\shortminus1$ & $\shortminus2,\shortminus1$ & \textcolor{gray}{$\shortminus1,\shortminus2$} & \textcolor{gray}{$\shortminus2,\shortminus1$}  \\
\hline
(1,0) & $\mathbin{\phantom{\shortminus}}0,\shortminus1$ & $\shortminus1,\shortminus1$ & $\shortminus1,\shortminus1$ & $\shortminus2,\shortminus1$ & \textcolor{gray}{$\shortminus2,\shortminus1$}  & \textcolor{gray}{$\shortminus1,\shortminus2$} \\
\hline
(1,1) & $\mathbin{\phantom{\shortminus}}0,\shortminus2$ & $\shortminus1,\shortminus2$ & $\shortminus1,\shortminus2$ & $\shortminus2,\shortminus2$ & \textcolor{gray}{$\shortminus2,\shortminus2$} & \textcolor{gray}{$\shortminus2,\shortminus2$} \\
\hline
\textcolor{gray}{(0,2)} & \textcolor{gray}{$\shortminus1,\shortminus1$} & \textcolor{gray}{$\shortminus1,\shortminus2$} &  \textcolor{gray}{$\shortminus2,\shortminus1$} & \textcolor{gray}{$\shortminus2,\shortminus2$} & \textcolor{gray}{$\shortminus2,\shortminus2$} & \textcolor{gray}{$\shortminus2,\shortminus2$} \\
\hline
\textcolor{gray}{(2,0)} & \textcolor{gray}{$\shortminus1,\shortminus1$} &  \textcolor{gray}{$\shortminus2,\shortminus1$} & \textcolor{gray}{$\shortminus1,\shortminus2$} & \textcolor{gray}{$\shortminus2,\shortminus2$} & \textcolor{gray}{$\shortminus2,\shortminus2$} & \textcolor{gray}{$\shortminus2,\shortminus2$}
\end{tabularx}
\end{center}
\caption{The payoff matrix of game $\mathcal{B}_{\$}^*$.}
\label{fig:payoff_matrix}
\end{figure}
If the opponent's equilibrium strategy is fixed, the player cannot change her payoff by deviating from one strategy in $S^*$ to another strategy in $S^*$ (this must be the case as every strategy in $S^*$ is an equilibrium strategy). However, by changing her strategy, the player can influence her opponent's equilibrium payoff.
\end{example}

\section{Equilibria Computation}
This section explains how every instance of $\mathcal{B}_{0}$ can be expressed as a variant of the Colonel Blotto game with fixed budgets under a well-defined strategy mapping. 
As a consequence, the problem of computing an equilibrium strategy of a player in the Colonel Blotto game with costs can be expressed as a linear program with a polynomial (in the number of battlefields and resources of both players) number of both constraints and variables using the BDDHS\footnote{The name of the linear program comes from the first letters of authors surnames.} linear program of Behnezhad et al.~(\citeyear{fast_2023_blotto}).

\subsection{Colonel Blotto Game with Sunk Costs}
In the classic Colonel Blotto game, a strategy of all player $P \in \{\playerA, \playerB\}$ is an assignment of all of her available $D^P$ resources over the set of the $[n]$ battlefields. Hence, the strategy set of a player $P$ is
\begin{equation}\label{def:pure_strategy_set_Plassical}
    M(D^P, n) = \left\{ \pure \in \mathbb{N}^n : \sum^n_{i=1} s_i = \thisRes \right\},
\end{equation}
which constitutes a subset of the player's strategy set $S(D^P, n)$ of the Colonel Blotto game with costs. Elements of set $M(D^P, n) $ are called full assignments, as players are required to assign all of their available resources. When the parameters are clear from the context, the set of full assignments is denoted $M^P$.

\begin{definition}
The Colonel Blotto game with costs, $\mathcal{B}(\aRes, \bRes, n, \bm{v}, \bm{0},\bm{0}, \bm{0}, \bm{0})$, where all costs functions are trivial (equal to $\bm{0})$ and strategy sets of both players are restricted to full assignments~\eqref{def:pure_strategy_set_Plassical} is called the \textbf{Colonel Blotto game with sunk costs} and denoted $\mathcal{B}(\aRes, \bRes, n, \bm{v}) \equiv \mathcal{B}$.
\end{definition}

 \subsection{Expressing $\mathcal{B}_{0}$ as $\mathcal{B}$}
For a given game $\mathcal{B}_{0}(\aRes, \bRes, n, \bm{v},\bm{c}^{\playerA}, \bm{c}^{\playerB}, g^{\playerA}, g^{\playerB})$, consider the \textit{corresponding} Colonel Blotto game with sunk costs $\mathcal{B}(\aRes, \bRes, n + 1, \hat{\bm{v}})$, with valuation functions $\hat{v}_i : \mathbb{N} \times \mathbb{N} \longrightarrow  \mathbb{R}$ defined as
\begin{align}\label{eq:payoff_zero-sum}
    & \hat{v}_i(s_i^{\playerA}, s_i^{\playerB}) = \\
    & 
    \begin{cases}
        v_i\left(s_i^{\playerA}, s_i^{\playerB}\right) - c_i^{\playerA}\left(s_i^{\playerA}\right) + c_i^{\playerB}\left(s_i^{\playerB}\right) \textnormal{, if $i \in \{1,2,\ldots,n\}$}, \\
        - g^{\playerA}\left(D^{\playerA} - s_{n+1}^{\playerA}\right) + g^{\playerB}\left(D^{\playerB} - s_{n+1}^{\playerB}\right) \textnormal{, if $i = n+1$.}
    \end{cases} \nonumber
\end{align}
The following result states that the Colonel Blotto game with costs and the corresponding Colonel Blotto game with sunk costs are strategically equivalent.

\begin{proposition}\label{prop:same_game}
For a given a tuple of parameters, the game $\mathcal{B}_{0}(\aRes, \bRes, n, \bm{v},\bm{c}^{\playerA},\bm{c}^{\playerB}, g^{\playerA}, g^{\playerB})$ and the corresponding Colonel Blotto game with sunk costs $\mathcal{B}(\aRes, \bRes, n + 1, \hat{\bm{v}})$ are the same game under the bijective strategy mapping $\bm{m} : S(D^P, n) \longrightarrow M(D^P, n+1)$, such that 
\begin{equation}\label{eq:strategy_mapping}
    \bm{m}(\bm{s}^P) = \left( s^P_1, s^P_2, \ldots, s^P_n, D^P - \sum_{i=1}^n s_i^P\right).
\end{equation}
\end{proposition}

\begin{proof}
For every strategy profile $(\bm{s}^{\playerA}, \bm{s}^{\playerB})$ it holds that
\begin{align*}
    &\pi_{0}^{\playerA} \left(\bm{s}^{\playerA}, \bm{s}^{\playerB}\right) \\
    = &\sum_{i=1}^n \left( v_i\left(s_i^{\playerA}, s_i^{\playerB}\right) - c_i^{\playerA} \left(s_i^{\playerA} \right) +  c_i^{\playerB} \left(s_i^{\playerB} \right)\right) \\
    & \;\;\; - g^{\playerA} \left( \sum_{i=1}^n s_i^{\playerA} \right) + g^{\playerB} \left( \sum_{i=1}^n s_i^{\playerB} \right) \\ 
    = &\sum_{i=1}^n \left( v_i\left(s_i^{\playerA}, s_i^{\playerB}\right) - c_i^{\playerA} \left(s_i^{\playerA} \right) +  c_i^{\playerB} \left(s_i^{\playerB} \right)\right) \\
    & \;\;\; - g^{\playerA} \left( D^{\playerA} - m_{n+1} \left( \bm{s}^{\playerA} \right) \right) + g^{\playerB} \left( D^{\playerB} - m_{n+1} \left( \bm{s}^{\playerB} \right) \right) \\
    = &\sum_{i=1}^{n+1} \hat{v}_i\left(m_i\left(\bm{s}^{\playerA} \right), m_i \left( \bm{s}^{\playerB} \right) \right) \\
    = & \; \; \pi^{\playerA} \left( \bm{m} \left( \bm{s}^{\playerA} \right), \bm{m} \left( \bm{s}^{\playerB} \right) \right),
\end{align*}
where $\pi^{\playerA} ( \bm{m} ( \bm{s}^{\playerA} ), \bm{m} \left( \bm{s}^{\playerB} \right))$ is the payoff to player $\playerA$ in the Colonel Blotto game with sunk costs $\mathcal{B}(\aRes, \bRes, n + 1, \hat{\bm{v}})$ under strategy profile $( \bm{m} ( \bm{s}^{\playerA} ), \bm{m} \left( \bm{s}^{\playerB} \right)) \in M(D^{\playerA}, n+1) \times M(D^{\playerB}, n+1)$. The payoff equality in both games holds for player $\playerB$ as well (the reason for using player $\playerA$ is purely technical, as it determines the sign of $v_i$ functions). Because of that, the two games are equivalent under the considered strategy mapping. 
\end{proof}
We conclude this subsection with the following observation regarding the relation between the strategies of the Colonel Blotto game with costs and the strategies of the corresponding Colonel Blotto game with sunk costs.
\begin{observation}\label{obs:obtained_resources}
Consider the Colonel Blotto game with costs, $\mathcal{B}_{\$}(\aRes, \bRes, n, \bm{v},\bm{c}^{\playerA},\bm{c}^{\playerB}, g^{\playerA}, g^{\playerB}) \equiv \mathcal{B}_{\$}$, and the corresponding Colonel Blotto game with sunk costs, $\mathcal{B}(\aRes, \bRes, n + 1, \hat{\bm{v}}) \equiv \mathcal{B}$. Fix Player $\thisPlayer$ and her pure strategy $\bm{s}^{\thisPlayer} \in M(D^{\thisPlayer}, n+1)$ in $\mathcal{B}$. The total number of resources obtained by player $\thisPlayer$ according to the corresponding pure strategy $\bm{m}^{\shortminus 1}(\bm{s}^{\thisPlayer}) \in S (D^{\thisPlayer}, n)$ of $\mathcal{B}_{\$}$ is exactly $D^{\thisPlayer} - s_{n+1}^{\thisPlayer}$.
\end{observation}

\subsection{Computational Complexity}
As a consequence of Propositions~\ref{prop:equiv} and~\ref{prop:same_game}, the problem of computing an equilibrium strategy of the Colonel Blotto game with costs can be restated as a problem of computing an optimal strategy of the corresponding Colonel Blotto game with sunk costs with one additional battlefield and the same numbers of resources of both players.

\begin{proposition}\label{prop:computing}
The problem of computing a player equilibrium strategy of the Colonel Blotto game with costs $\mathcal{B}_{\$}(\aRes, \bRes, n, \bm{v}, \bm{c}^{\playerA},\bm{c}^{\playerB}, g^{\playerA}, g^{\playerB})$ can be formulated as a linear program with at most $\Theta ((D^{\max})^2 \cdot n)$ constraints and $\Theta ((D^{\max})^2 \cdot n)$ variables, where $D^{\max} = \max\{D^{\playerA}, D^{\playerB}\}$.
\end{proposition}

\begin{proof}
Consider an instance of the Colonel Blotto game with costs, $\mathcal{B}_{\$}(\aRes, \bRes, n, \bm{v}, \bm{c}^{\playerA},\bm{c}^{\playerB}, g^{\playerA}, g^{\playerB})$. By Proposition~\ref{prop:equiv}, computing an equilibrium strategy of $\mathcal{B}_{\$}$ is equivalent to computing optimal strategy of the corresponding game $\mathcal{B}_{0}(\aRes, \bRes, n, \bm{v}, \bm{c}^{\playerA},\bm{c}^{\playerB}, g^{\playerA}, g^{\playerB})$. By Proposition~\ref{prop:same_game}, this is equivalent to computing an optimal strategy of the corresponding Colonel Blotto game with sunk costs, $\mathcal{B}(\aRes, \bRes, n + 1, \hat{\bm{v}})$.
By Theorem~1~\cite{fast_2023_blotto}, an optimal strategy of the Colonel Blotto game with sunk costs, $\mathcal{B}(\aRes, \bRes, n + 1, \hat{\bm{v}})$ can be expressed as a linear program with at most $\Theta ((D^{\max})^2 \cdot (n+ 1))$ constraints and $\Theta ((D^{\max})^2 \cdot (n+1))$ variables, where $D^{\max} = \max\{D^{\playerA}, D^{\playerB}\}$. Since $\Theta ((D^{\max})^2 \cdot (n+ 1)) = \Theta ((D^{\max})^2 \cdot n)$, this concludes the proof.
\end{proof}
As every pair of equilibrium strategies is a Nash equilibrium of the Colonel Blotto game with costs (Corollary~\ref{cor:interchangeability}), the following theorem is a direct consequence of Proposition~\ref{prop:computing}.
\begin{theorem}
A Nash equilibrium of the Colonel Blotto game with costs can be computed in polynomial time with respect to the number of battlefields and the number of resources of both players.
\end{theorem}

We conclude this section with the high-level idea of our implementation of a program computing an equilibrium strategy of the Colonel Blotto game with costs. 
In the description of our implementations we refer to the linear program BDDHS proposed by Behnezhad et al.~(\citeyear{fast_2023_blotto}), which when given an instance of the Colonel Blotto game with sunk costs, $\mathcal{B}$, defines a set of variables and constraints capturing any mixed strategy of a player. For clarity and conciseness, we only refer to the variables we operate on directly:
\begin{itemize}
    \item $u^{\playerB}$ - a variable of the program describing the bound on the players' $\playerB$ maximal payoff in $\mathcal{B}_0(n, D^{\playerA}, D^{\playerB}, \bm{v})$;
    \item $P_i(t)$ - a variable describing the probability of player $\playerA$ assigning exactly $t$ resources to battlefield $i$.
\end{itemize}
Algorithm~\ref{alg:lp_0} describes the idea of our implementation.
\begin{algorithm}
\caption{Computing an equilibrium strategy of player $\playerA$ of the Colonel Blotto game with costs.}\label{alg:lp_0}
\textbf{Input}: A set of parameters $(\aRes, \bRes, n, \bm{v}, \bm{c}^{\playerA},\bm{c}^{\playerB}, g^{\playerA}, g^{\playerB})$. \\
\textbf{Output}: Linear program describing an equilibrium strategy of player $\playerA$.
\begin{algorithmic}[1]
\STATE LP = BDDHS($\mathcal{B}(\aRes, \bRes, n + 1, \hat{\bm{v}})$) \\ \;\;\# Payoff function vector $\hat{\bm{v}}$ is defined as in \\ \;\;\# Equation~\eqref{eq:payoff_zero-sum}.
\STATE LP.minimize($u^{\playerB}$)
\RETURN LP
\end{algorithmic}
\end{algorithm}
Formally, Algorithm~\ref{alg:lp_0} returns a linear program describing an optimal strategy of the Colonel Blotto game with sunk costs $\mathcal{B}(\aRes, \bRes, n + 1, \hat{\bm{v}})$. However, the inverse $m^{\shortminus 1}$ of strategy mapping $m$ defined in Equation~\eqref{eq:strategy_mapping} defines the exact relation between any optimal strategy described by LP and the corresponding equilibrium strategy of the Colonel Blotto game with costs.
As stated by Observation~\ref{obs:obtained_resources}, for a given number of resources $t \in \{0,1,\ldots,D^{\playerA}\}$, the variable $P_{n+1}(t)$ in the resulting LP program describes the probability that player $\playerA$ obtained exactly $(D^\playerA - t)$ resources in the equilibrium strategy.

\section{Experimental Results}
In this section, we report on the carried-out experiments, focusing on the amount of resources players obtain in equilibrium. First, we provide a high-level description of the implemented program, used for finding the lower (upper) bound on the number of resources obtained in equilibrium. Second, we report and discuss the obtained results. Lastly, we report the running times of our program.

In order to investigate the number of resources the players obtain in equilibrium, we introduced a new variable, $D$, which describes the expected number of resources obtained by player $\playerA$. Algorithm~\ref{alg:lp_max} describes the high-level idea of computing the maximal (minimal) number of resources obtained by a player in equilibrium.

\begin{algorithm}
\caption{Computing the maximal (minimal) expected number of resources in any equilibrium strategy of player $\playerA$ of the Colonel Blotto game with costs.}\label{alg:lp_max}
\textbf{Input}: A set of parameters $(\aRes, \bRes, n, \bm{v}, \bm{c}^{\playerA},\bm{c}^{\playerB}, g^{\playerA}, g^{\playerB})$. \\
\textbf{Output}: Maximal (minimal) number of resources obtained by player $\playerA$ in equilibrium.
\begin{algorithmic}[1]
\STATE LP = BDDHS($\mathcal{B}(\aRes, \bRes, n + 1, \hat{\bm{v}})$) \\ \;\;\# Payoff function vector $\hat{\bm{v}}$ is defined as in \\ \;\;\# Equation~\eqref{eq:payoff_zero-sum}.
\STATE LP.minimize($u^{\playerB}$)
\STATE LP.fixVariableValue($u^{\playerB}$) 
\\ \;\; \# LP is restricted to equilibrium strategies.
\STATE LP.addVariable(D)
\STATE LP.addConstraint(D $= \sum_{t = 0}^{D^{\playerA}} P_{n+1}(t) \cdot (D^{\playerA} - t)$) 
\\ \;\;\# D describes the expected number of resources 
\\ \;\;\# obtained by player $\playerA$.
\STATE LP.maximize(D)
\RETURN D
\end{algorithmic}
\end{algorithm}

We proceed to report our results. Specifically, we investigate the relation between the costs of obtaining resources and the number of resources players obtain in equilibrium of the Colonel Blotto game with costs, using the implementation of Algorithm~\ref{alg:lp_max}. We consider models with linear costs, $\mathcal{B}_{\$}$, such that
\begin{align*}
    & D^{\playerA} = D^{\playerB} \text{, and } \\
    & \pi^{\thisPlayer} (\thisPure, \thatPure) = \sum_{i = 1}^n \sign(s_i^{\thisPlayer} - s_i^{\thatPlayer}) - c_0 \cdot \sum_{i=1}^n s_i^{\thisPlayer},
\end{align*}
for some unit costs coefficient $c_0$.
The following hypothesis captures the structure of the preliminary results of our experiments.
\begin{hypothesis}\label{obs:structure}
All of the results concerning the equilibrium number of resources, denoted by $D(D^{\playerA}, n, c)$, satisfy
\begin{align*}
    & D\left(D^{\playerA}, n, c_0\right) = D^{\playerA} \text{, if } D^{\playerA} \leq n \cdot \left(\left\lfloor \frac{1}{c_0} \right\rfloor - 1\right), \\
    & D\left(D^{\playerA}, n, c_0\right) = \min\left\{D^{\playerA}, n \cdot \left\lfloor \frac{1}{c_0} \right\rfloor \right\} \text{, if } \left\lfloor \frac{1}{c_0} \right\rfloor < \frac{1}{c_0}, \\
    & D\left(D^{\playerA}, n, c_0\right) \in \left(n \cdot \left(\frac{1}{c_0} - 1 \right), \min\left\{n \cdot \frac{1}{c_0}, D^{\playerA}\right\}\right), \\
    & \text{ otherwise, i.e., when } \\
    & D^{\playerA} > n \cdot \left(\left\lfloor \frac{1}{c_0} \right\rfloor - 1\right) \text{ and } \left\lfloor \frac{1}{c_0} \right\rfloor = \frac{1}{c_0}.
\end{align*}
\end{hypothesis}
The preliminary results suggest the following structure. The only case when the number of resources used in equilibrium is not unique (i.e., the lower and the upper bounds are different), is when the inverse of unit costs, $c_0^{\shortminus 1}$, is a natural number and the number of resources, $D^{\playerA}$ is sufficiently large. Since the set of equilibrium strategies is convex, for such cases every number of resources between the lower and upper bound is also obtained in some equilibrium strategy. In order to test Hypothesis~\ref{obs:structure}, we ran experiments operating on the inverse of the unit costs as a parameter. Figure~\ref{fig:params_2} presents the considered range of parameters.
\begin{figure}[H]
\begin{center}
\begin{tabularx}{0.95\columnwidth}{ 
   >{\centering\arraybackslash}X 
   >{\centering\arraybackslash}X 
   >{\centering\arraybackslash}X 
   >{\centering\arraybackslash}X 
   >{\centering\arraybackslash}X  }
\hline
 & $n$ & $D_A$ & $D_B$ & $c_0^{\shortminus 1}$ \\
\hline
min & 3 & 1 & 1 & 1 \\
max & 5 & 50 & 50 & 10 \\
interval & 1 & 1 & 1 & 0.25  \\
\hline
\end{tabularx}
\end{center}
    \caption{The range of parameters considered in the experiments aimed to test Hypothesis~\ref{obs:structure}. The min (max) row corresponds to the lower (upper) bound on the considered parameter, and the interval row describes the difference between the two consecutive values of the parameter.}
    \label{fig:params_2}
\end{figure}
The results fit Hypothesis~\ref{obs:structure}. Figure~\ref{fig1} illustrates the maximal and minimal number of resources used by a player in equilibrium as a function of the inverse of the unit costs, $c_0^{\shortminus 1}$, for $c_0^{\shortminus 1} \in \{1, \frac{5}{4}, \frac{6}{4}, \ldots, 10\}$.
\begin{figure}
\centering
\includegraphics[width=0.8\columnwidth]{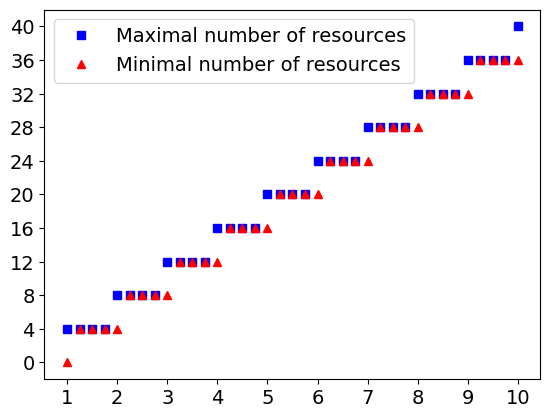} 
\caption{The $y$-axis is the number of resources used in equilibrium and the $x$-axis is the inverse $c_0^{\shortminus 1}$ of the unit costs, $c_0$. The considered number of available resources of both players is $D^{\playerA} = 40$ and players compete on $n=4$ battlefields.}
\label{fig1}
\end{figure}
As illustrated in Figure~\ref{fig1}, the number of resources obtained in equilibrium increases when the unit cost decreases. The intuition behind it is that the same total carried cost allows players to obtain more resources as the unit cost decreases.

We look at the total costs carried by a player in equilibrium and check the idea of convergence to the continuous model. Figure~\ref{fig2} presents the maximal and minimal equilibrium expenditure for the same experimental data that are presented in Figure~\ref{fig1}. Figure~\ref{fig2} also includes a constant line presenting the unique equilibrium expenditure of the corresponding continuous model of Roberson and Kvasov~(\citeyear{kvasov}).

\begin{figure}
\centering
\includegraphics[width=0.8\columnwidth]{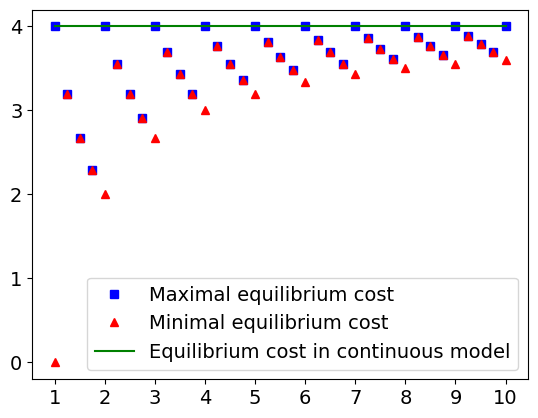} 
\caption{The $y$-axis is the cost carried in an equilibrium and the $x$-axis is the inverse $c_0^{\shortminus 1}$ of the unit costs, $c_0$.}
\label{fig2}
\end{figure}

We move on to report on the running times of our implementation of Algorithm~\ref{alg:lp_0}. The BDDHS linear program used in Algorithm~\ref{alg:lp_0} was implemented as presented by Behnezhad et al.~(\citeyear{fast_2023_blotto}). Our implementation was done in $C\texttt{++}$, using the state-of-the-art \textit{Gurobi}~\cite{gurobi} LP-optimizer.
We investigated the program running time for two settings. In the first setting, we considered linear obtainment costs with unit costs of a resource equal to $0.05$ and no assignment costs. Hence, the payoff $\hat{\pi}$ to player $\thisPlayer$ is
\begin{equation*}
    \hat{\pi}^{\thisPlayer} (\thisPure, \thatPure) = \sum_{i = 1}^n \sign(s_i^{\thisPlayer} - s_i^{\thatPlayer}) - \frac{1}{20} \cdot \sum_{i=1}^n s_i^{\thisPlayer}.
\end{equation*}
In the second setting, we considered squared assignment costs with a coefficient of 0.01 and no obtainment costs. Hence, the payoff $\bar{\pi}$ to player $\thisPlayer$ is
\begin{equation*}
    \bar{\pi}^{\thisPlayer} (\thisPure, \thatPure) = \sum_{i = 1}^n \sign(s_i^{\thisPlayer} - s_i^{\thatPlayer}) - \frac{1}{100} \cdot \sum_{i=1}^n (s_i^{\thisPlayer})^2.
\end{equation*}
Figure~\ref{tab:time} presents the obtained running times. We ran the code on a machine with an Apple M3 Pro processor unit and an 18-GB memory. 
\begin{figure}[H]
\begin{center}
\begin{tabularx}{0.95\columnwidth}{ 
   >{\centering\arraybackslash}s
   >{\centering\arraybackslash}s 
   >{\centering\arraybackslash}s 
   >{\centering\arraybackslash}b 
   >{\centering\arraybackslash}b  }
\hline
$n$ & $D_A$ & $D_B$ & Time ($\hat{\pi}$)  & Time ($\bar{\pi}$)\\
\hline
10 & 30 & 30 & 0.534s & 0.486s \\
10 & 30 & 50 & 2.114s & 1.629s \\
10 & 50 & 50 & 3.964s & 4.110s \\
20 & 30 & 30 & 1.296s & 1.615s \\
20 & 30 & 50 & 8.070s & 6.739s \\
20 & 50 & 50 & 10.287s & 10.968s \\
30 & 30 & 30 & 3.354s & 2.180s \\
30 & 30 & 50 & 15.177s & 16.376s \\
30 & 50 & 50 & 15.182s & 19.988s \\
\hline
\end{tabularx}
\end{center}
\caption{The averages of 10 running times of our implementation of Algorithm~\ref{alg:lp_0}. The fourth column ($\hat{\pi}$) presents the running times for models with linear obtainment costs, while the fifth column ($\bar{\pi}$) presents the running times for models with squared assignment costs.}
\label{tab:time}
\end{figure}
\section{Conclusions}
In this paper, we studied the Colonel Blotto game with costs. We demonstrated that the equilibrium strategies of this generalized game can be computed in polynomial time with respect to the game parameters. Furthermore, we provided a comprehensive method for computing these equilibrium strategies.
The significance of our contribution is two-fold. Firstly, the descriptiveness of the studied model, which allows for capturing multiple real-life scenarios, makes our method for computing equilibrium strategies an excellent candidate for integration into decision-support systems. Secondly, the structured nature of the experimental results offers a promising avenue for testing hypotheses regarding the properties of equilibria, which may aid in obtaining further theoretical insights into the model.

\section{Acknowledgments}
The author thanks Marcin Dziubiński for the discussions and his valuable remarks. Stanisław Kaźmierowski is supported by the European Union (ERC, PRO-DEMOCRATIC, 101076570). Views and opinions expressed are however those of the author only and do not necessarily reflect those of the European Union or the European Research Council. Neither the European Union nor the granting authority can be
held responsible for them.

\centering
\includegraphics[width=0.7\columnwidth]{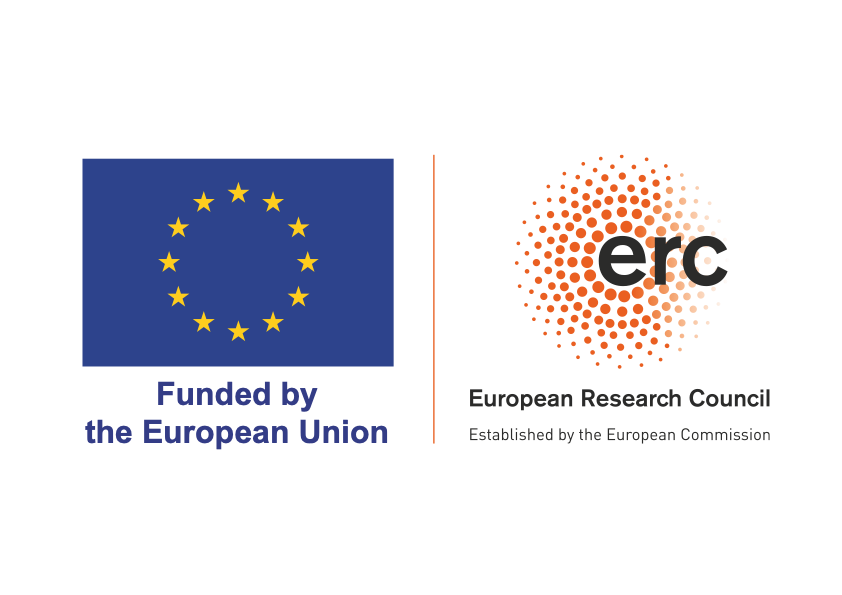}

\bibliography{paper}

\end{document}